\documentclass[11pt]{article}

\usepackage[a4paper,margin=1in]{geometry}

\usepackage[T1]{fontenc}
\usepackage[utf8]{inputenc}
\usepackage{lmodern}
\usepackage{microtype}
\usepackage{amsmath,amssymb,amsfonts,mathtools}
\usepackage{graphicx}
\usepackage{booktabs}
\usepackage{csquotes}
\usepackage{hyperref}
\usepackage[capitalise]{cleveref}
\usepackage[british]{babel}
\usepackage{array}
\usepackage{caption}
\usepackage{graphicx}
\usepackage{ragged2e}
\usepackage{float}
\usepackage[inline]{enumitem}
\usepackage{comment}
\usepackage{pdflscape}
\usepackage{soul}
\usepackage{xcolor}
\usepackage{amssymb}
\usepackage{algpseudocode}  
\usepackage{amsthm}                
\usepackage[framemethod=default]{mdframed} 
\usepackage{xcolor}                
\usepackage{placeins}   
\usepackage{array} 
\newcolumntype{L}[1]{>{\raggedright\arraybackslash}p{#1}}
\usepackage{accents}

\newtheorem{theorem}{Theorem}[section]
\newtheorem{lemma}[theorem]{Lemma}
\newtheorem{corollary}[theorem]{Corollary}

\theoremstyle{definition}
\newtheorem{innerdefinition}{Definition}[section]

\newmdenv[
skipabove=\baselineskip,
skipbelow=\baselineskip,
linewidth=0.6pt,
linecolor=black,
backgroundcolor=gray!5, 
roundcorner=2pt
]{defbox}

{\begin{defbox}\begin{innerdefinition}[#1]}%
		{\end{innerdefinition}\end{defbox}}

\usepackage[ruled,vlined,linesnumbered]{algorithm2e}
\SetKwInput{KwPre}{Preconditions}
\SetKwInput{KwPost}{Postcondition}
\SetKw{Return}{return}
\SetKwInOut{Init}{Init}

\usepackage{tikz}
\usetikzlibrary{arrows.meta,positioning,fit,shapes.geometric}

\usepackage{subcaption} 

\usepackage{amsmath,amssymb,mathtools}   
\usepackage{booktabs}                    
\usepackage{array}                       

\newcommand{\EqV}{\mathsf{eq\_vertex}}
\newcommand{\EqHs}{\mathsf{eq\_hs}}

\newcommand{\SameName}{\texttt{SameName}}

\newcommand{\nar}{\textit{n}-ary}

\usepackage{amsmath,amssymb,mathtools}

\newcommand{\medtilde}[1]{\stackrel{\sim}{#1}}

\usepackage[acronym]{glossaries}
\makeglossaries
\usepackage{pdflscape}

\title{Hypernetwork Theory: The Structural Kernel}
\author{Richard D. Charlesworth\thanks{Visiting Research Fellow, School of Engineering and Innovation, The Open University, UK}}
\date{\today}

\begin{document}
	\maketitle

	\begin{abstract}
		Modelling across engineering, systems science, and formal methods remains limited by binary relations, implicit semantics, and diagram-centred notations that obscure multilevel structure and hinder mechanisation. Hypernetwork Theory (HT) addresses these gaps by treating the $n$-ary relation as the primary modelling construct. Each relation is realised as a typed hypersimplex---$\alpha$ (conjunctive, part–whole) or $\beta$ (disjunctive, taxonomic)---bound to a relation symbol $R$ that fixes arity and ordered roles. Semantics are embedded directly in the construct, enabling hypernetworks to represent hierarchical and heterarchical systems without reconstruction or tool-specific interpretation.
		
		This paper presents the structural kernel of HT. It motivates typed $n$-ary relational modelling, formalises the notation and axioms (A1–A5) for vertices, simplices, hypersimplices, boundaries, and ordering, and develops a complete algebra of structural composition. Five operators---merge, meet, difference, prune, and split---are defined by deterministic conditions and decision tables that ensure semantics-preserving behaviour and reconcile the Open World Assumption with closure under rules. Their deterministic algorithms show that HT supports reproducible and mechanisable model construction, comparison, decomposition, and restructuring.
		
		The resulting framework elevates hypernetworks from symbolic collections to structured, executable system models, providing a rigorous and extensible foundation for mechanisable multilevel modelling.
	\end{abstract}
	

\section{Introduction}
\label{sec:introduction}

Across systems engineering, computing, and formal methods, modelling practice remains dominated by diagrams. In notations such as \textit{Unified Modelling Language} (UML) and \textit{Systems Modelling Language} (SysML)~\cite{delligatti2013sysml,eriksson2003uml}, the diagram is often mistaken for the model itself. Diagrams communicate structure, but the model is a structured, semantically grounded specification that can be rendered in many forms. When modelling is reduced to diagramming, semantics remain implicit, precision weakens, and mechanisation becomes unreliable. Hypernetwork Theory (HT) adopts a \emph{model-first} stance: diagrams are views; the model is the underlying semantics.

HT specialises this stance by treating the $n$-ary relation as fundamental. Each relation is realised as a typed hypersimplex---$\alpha$ (conjunctive, part--whole) or $\beta$ (disjunctive, taxonomic)---bound to a relation symbol $R$ that fixes arity and ordered roles. This embeds semantics directly in the structure, avoiding the untyped nature of graphs and hypergraphs~\cite{RN3,RN363,RN777} and the triple-based reconstruction required in ontology frameworks such as \textit{Resource Description Framework} (RDF) and the \textit{Web Ontology Language} (OWL)~\cite{RN688,brickley2004rdf}. Typed $n$-ary relations address long-standing challenges in systems modelling where part--whole, taxonomic, and associative semantics must coexist without collapse.

HT reconciles extensibility with semantic rigour. Under the Open World Assumption~\cite{Keet2013a}, hypernetworks can be freely extended; at the same time, every construct must satisfy axioms governing identity, typing, relation binding, and boundary scoping. Openness applies to scope, while closure applies to semantics. This contrasts with \textit{Meta-Object Facility} (MOF)-based languages~\cite{siegel2014object} and \textit{Model Driven Engineering} (MDE)/\textit{Model Driven Architecture} (MDA) approaches~\cite{RN709,gavsevic2009model,omg2014mda}, which assume a Closed World stance and rely on tool conventions for meaning.

Most established approaches enforce hierarchy: UML and SysML rely on containment~\cite{delligatti2013sysml,eriksson2003uml}; MOF prescribes a fixed four-layer stack~\cite{siegel2014object}; and ontologies impose taxonomic stratification~\cite{brickley2004rdf,RN688}. These assumptions break down when systems exhibit heterarchy, cross-level relations, or overlapping subsystems. HT instead treats hypersimplices as heterarchical constructs: any hypersimplex may contain vertices, simplices, or other hypersimplices, regardless of ``level''. Boundaries provide scoped visibility without altering identity, supporting overlapping subsystems while maintaining global consistency.

This paper presents the structural kernel of Hypernetwork Theory as a unified whole, consolidating and extending the formal developments first outlined in the author’s doctoral thesis~\cite{charlesworth2025thesis}. It provides a clarified and mechanisable account of typed $n$-ary relational modelling, formalising refined notation and axioms (A1--A5) for vertices, simplices, hypersimplices, relation binding, aggregation typing, boundaries, and ordering. It develops a complete algebra of structural operators---merge, meet, difference, prune, and split---together with validity conditions (C1--C7) and deterministic decision tables, forming a semantics-preserving and mechanisable modelling framework. Boundaries are introduced only to the extent required for notation and axioms; their full structural behaviour is developed in a companion paper.

\paragraph{Contributions.}
This unified paper makes four contributions:
\begin{enumerate}[leftmargin=*]
	\item \textbf{Conceptual foundations.} A model-first account of explicit, typed $n$-ary relational structure in which diagrams are views of underlying semantics rather than the model itself.
	\item \textbf{Refined notation and axioms.} A precise definition of vertices, simplices, hypersimplices, aggregation typing, relation binding, boundaries, and ordering as semantically explicit constructs governed by axioms~A1--A5.
	\item \textbf{Operational semantics.} A complete algebra of structural operators---merge, meet, difference, prune, and split---supported by validity conditions~C1--C7 and decision tables that ensure deterministic, semantics-preserving composition.
	\item \textbf{Mechanisation.} Deterministic algorithms realising the operator algebra, demonstrating that HT is fully executable and providing a reproducible foundation for tooling.
\end{enumerate}
	
\paragraph{Organisation.}
Section~\ref{sec:background-motivation} outlines the conceptual motivation and limitations of existing approaches. Section~\ref{sec:notation} formalises the refined notation and constructs. 
Section~\ref{sec:axioms} presents axioms~A1--A5. Section~\ref{sec:structure-and-composition} explains the structural properties that follow from these axioms. Section~\ref{sec:axioms-validity} introduces operator-level validity conditions~C1--C7. Section~\ref{sec:algebraic} develops the algebra of merge, meet, difference, prune, and split. Section~\ref{sec:mechanisation} details their mechanisation. Section~\ref{sec:examples} provides worked examples. Section~\ref{sec:discussion} discusses generality, limitations, and implications. Section~\ref{sec:related} positions the work within related formalisms. Section~\ref{sec:conclusion} concludes and identifies the boundary calculus as the next stage in the development of HT.

\section{Background and Motivation: Why Existing Approaches Fall Short}
\label{sec:background-motivation}

Modelling approaches span graphs and hypergraphs, UML and SysML, MOF-based languages, ontology frameworks such as RDF/OWL, and Model-Driven Engineering (MDE). Despite their diversity, four limitations recur: $n$-ary relations are absent or reconstructed indirectly; multilevel structure is imposed by convention rather than intrinsic constructs; composition is ad hoc or tool-dependent; and semantics remain implicit. Abstract traditions---including category theory~\cite{goguen1991manifesto,spivak2014ct} and sheaf theory~\cite{curry2014sheaves,maclane1992sheaves}---address composition and locality in general form but do not yield an executable, semantics-preserving foundation.

\subsection{Graphs, Hypergraphs, Networks of Networks, and Object Orientation}
Graphs capture only binary relations, forcing higher-order semantics into chains or auxiliary constructs. Hypergraphs extend to n\text{-}ary edges but treat them as untyped sets~\cite{RN3, RN383}, leaving aggregation semantics and role alignment external. Networks\text{-}of\text{-}networks and multilayer or multiplex models strengthen graph structure by adding layers or inter\text{-}layer couplings, but they remain graph\text{-}based and therefore binary at the semantic level, requiring n\text{-}ary structure, part--whole semantics, and taxonomic alternatives to be reconstructed indirectly~\cite{RN363,RN385,RN362,RN361}. Object orientation~\cite{booch2008object} provides class\text{-}level structure but conflates part--whole and taxonomic semantics, with composition largely informal, tool\text{-}specific, and driven by naming conventions rather than typed relational structure. HT differs by realising each n\text{-}ary relation as an ordered, typed hypersimplex bound to a relation symbol $R$, embedding $\alpha/\beta$ semantics directly in the construct.

\subsection{UML, SysML, MOF, and MDE}
UML and SysML~\cite{delligatti2013sysml,eriksson2003uml} centre on diagrammatic views and binary associations, with hierarchy enforced by the MOF stack~\cite{siegel2014object}. These notations prioritise communication but rely on implicit semantics and operate under the Closed World Assumption~\cite{Keet2013b}. MDE introduces transformation pipelines such as \textit{Atlas Transformation Language} (ATL) and \textit{Query/View/Transformation} (QVT)~\cite{jouault2008atl,kurtev2006qvt,omg2016qvt} and tooling such as \textit{Eclipse Modeling Framework} (EMF)~\cite{steinberg2008emf}, but inherits binary meta-models and lacks semantics-preserving guarantees. Composition remains brittle and tool-dependent.

\subsection{Ontologies and Description Logics}
Ontologies in RDF and OWL~\cite{brickley2004rdf,RN688}, grounded in Description Logics~\cite{RN666,RN664}, reconstruct $n$-ary relations via triples, mediating classes, or reified role patterns. This reconstruction is indirect and obscures aggregation semantics. Algebraic composition is absent, cross-level links are not first-class, and although these frameworks adopt the Open World Assumption~\cite{Keet2013a}, inference is required to determine exclusion, scope, and structural interpretation, limiting mechanisability.

\subsection{Abstract and Systems-Level Traditions}
Category theory~\cite{goguen1991manifesto,spivak2014ct} formalises composition abstractly but leaves typing, ordering, and scoping outside the framework. Sheaf theory~\cite{curry2014sheaves,maclane1992sheaves} provides local-to-global coherence but requires fixed structures that are difficult to mechanise. Systems-of-systems engineering~\cite{jamshidi2008sos,maier1998architecting} emphasises subsystem boundaries and coordination, but these boundaries are descriptive rather than structural and lack the algebraic closure and identity-preserving behaviour required for reproducible composition.

\begin{table}[h!]
	\centering
	\caption{Comparison of established modelling approaches with Hypernetwork Theory (HT).}
	\label{tab:comparison}
	\begin{tabular}{p{0.20\textwidth} p{0.37\textwidth} p{0.37\textwidth}}
		\hline
		\textbf{Approach} & \textbf{Limitations / Characteristics} & \textbf{HT Contribution} \\
		\hline
		Graphs & Binary-only; higher-order structure forced through chains; no embedded semantics. & Native typed $n$-ary relations with explicit $\alpha/\beta$ semantics. \\
		\hline
		Hypergraphs & $n$-ary edges untyped; semantics external. & Hypersimplices typed as $\alpha$ or $\beta$; semantics embedded directly. \\
		\hline
		Object Orientation & Conflates part--whole and taxonomic semantics; composition informal. & Explicit $\alpha/\beta$ distinction; semantics-preserving operators. \\
		\hline
		Ontologies & $n$-ary structure reconstructed; cross-level links absent; inference required. & Native typed $n$-ary relations; explicit exclusion; scoped boundaries; identity-preserving projection. \\
		\hline
		UML / SysML / MOF & Diagram-centric; binary associations; enforced hierarchy; CWA by default. & Model-first stance; heterarchy as first-class; openness with rule-based closure. \\
		\hline
		Model-Driven Engineering & Binary meta-models; composition via scripts; tool-dependent. & Deterministic structural operators enabling mechanisable composition. \\
		\hline
		Category Theory & Composition abstract; semantics external. & Typed relations, roles, ordering, and boundaries as structural primitives. \\
		\hline
		Sheaf Theory & Local–global consistency but fixed structure; hard to mechanise. & Algorithmic projection with rule-based closure. \\
		\hline
		Systems of Systems & Descriptive boundaries; no algebraic closure. & Identity-preserving scope and semantics-preserving operators. \\
		\hline
	\end{tabular}
\end{table}

\subsection{Summary of Gaps}
Five persistent gaps recur across established approaches:  
(i) $n$-ary relations remain implicit or reconstructed;  
(ii) multilevel structure is imposed rather than intrinsic;  
(iii) hierarchy is enforced rather than optional;  
(iv) composition is ad hoc or informal;  
(v) mechanisation depends on diagrams, inference, or tooling.

HT addresses these gaps by embedding semantics directly in its constructs---typed $n$-ary relations, explicit relation binding, scoped boundaries, ordering, and modeller-supplied exclusion---and by enforcing closure algorithmically. This reconciliation of openness with rule-based validity motivates the notation, axioms, and operator algebra developed in the remainder of this paper.

\section{Hypernetwork Notation and Core Constructs}
\label{sec:notation}

Hypernetwork Theory is due to Johnson~\cite{RN365,RN197,RN232}. This paper builds on his original formulation of vertices, simplices, hypersimplices, and $\alpha/\beta$ aggregations, extending the theory with refined notation, formal axioms, and a mechanisable algebra of structural operators. These constructs extend modelling beyond graphs and hypergraphs by giving $n$-ary relations explicit structure. The refinements below make these constructs formally precise and directly mechanisable, enabling HT to support typed $n$-ary semantics, multilevel structure, boundaries, and deterministic structural operations.

\subsection{Aggregation Types ($\alpha$, $\beta$)}
Every hypersimplex must declare its aggregation type. A minimal micro-example illustrates each case:
\[
\begin{aligned}
	\varsigma^{\alpha}_{CarAssembly} &= \langle \textit{Chassis}, \textit{Engine}, \textit{WheelSet} \; ; \; R_{carAssembly} \rangle,\\
	\varsigma^{\beta}_{VehicleType} &= \{ \textit{Car}, \textit{Van}, \textit{Truck}, \textit{SUV} \; ; \; R_{vehicleType} \}.
\end{aligned}
\]

The $\alpha$-aggregation expresses part–whole semantics: removing \emph{Engine} changes the identity of \emph{Car}. The $\beta$-aggregation expresses taxonomic alternatives: removing \emph{Van} does not change the identity of the remaining vehicle types.

\subsection{Vertices and Hypersimplices}

Vertices are atomic identifiers that act as the $0$-simplices of the theory. A \emph{simplex} is the standard topological construct: an ordered tuple of vertices (or, recursively, hypersimplices) representing an abstract $n$-simplex. Hypernetwork Theory extends this topological notion by attaching a relation symbol $R$ and an explicit $\alpha/\beta$ aggregation type to the simplex, yielding a \emph{hypersimplex}. The hypersimplex is therefore the semantic unit of the theory: it records the ordered participants, their role alignment under $R$, and whether the $n$-ary relation is conjunctive or disjunctive.

A minimal example illustrates the distinction:
\[
\sigma = \langle \textit{Rim, Hub, Tyre} \rangle
\qquad
\varsigma^{\alpha}_{\text{WheelAssembly}}
= \langle \textit{Rim, Hub, Tyre} ; R_{\text{wheelAssembly}} \rangle .
\]
Here $\sigma$ is the underlying topological simplex, while $\varsigma$ is the typed $n$-ary relation built from it.

Because hypersimplices may include vertices or other hypersimplices as ordered participants, they form a recursive hierarchy capable of expressing heterarchical and multilevel structure directly.

A modeller may also introduce an \emph{anti-vertex} $\tilde v$ to mark explicit exclusion, e.g.
\[
\langle \textit{Wheel}_1, \textit{Wheel}_2, \textit{Wheel}_3, \textit{Wheel}_4,
\medtilde{\textit{SpareWheel}} ; R_{\text{carWheels}} \rangle ,
\]
indicating that a spare wheel is intentionally absent. Anti-vertices are explicit constructs, introduced by the modeller or by operators, and do not affect closure, which is guaranteed by the axioms and operators of HT.

\subsection{Hypernetworks}

Johnson’s original formulation of hypernetworks~\cite{RN365,RN197,RN232} treats systems as structured aggregations of relations built from vertices and hypersimplices arranged through containment. In this classical view, a hypersimplex is an untyped higher–order relational unit: a topological simplex whose participants have been grouped to form a relation, with the intended semantics inferred rather than declared. This captures how higher–order and multilevel organisation arise once relations can contain relations, but the semantic nature of each aggregation—whether conjunctive, taxonomic, or alternative—remains implicit and cannot be checked or mechanised.

Hypernetwork Theory (HT) retains the insight that systems are structured relational compositions but makes the semantics explicit. Each $n$-ary relation is typed as $\alpha$ (conjunctive) or $\beta$ (disjunctive), participants are aligned to ordered roles under a relation symbol $R$, identifiers are preserved globally, boundaries provide scoped visibility, and deterministic operators act on these constructs. The originally implicit classical formulation is therefore sharpened into a precise, typed, and executable semantics.

A hypernetwork is simply a finite collection of typed $n$-ary relations, each realised as a hypersimplex. For example:
\[
\begin{aligned}
	H = \big\{
	&\varsigma^{\alpha}_{\text{Car}}
	= \langle \varsigma^{\alpha}_{\text{WheelAssembly}},
	\varsigma^{\alpha}_{\text{FrameSet}} ;
	R_{\text{car}} \rangle, \\[2pt]
	&\varsigma^{\alpha}_{\text{WheelAssembly}}
	= \langle \textit{Rim, Hub, Tyre} ;
	R_{\text{wheelAssembly}} \rangle, \\[2pt]
	&\varsigma^{\alpha}_{\text{FrameSet}}
	= \langle \textit{Frame, SeatTube, TopTube} ;
	R_{\text{frameSet}} \rangle, \\[2pt]
	&\textit{Rim, Hub, Tyre, Frame, SeatTube, TopTube}
	\big\}
\end{aligned}
\]
This $H$ is a hypernetwork: a coherent structure formed from vertices and typed hypersimplices.

\subsection{Notation}
This subsection fixes the concrete notation used throughout the paper, giving the symbols for vertices, anti-vertices, simplices, hypersimplices, relation binding, aggregation typing, boundaries, and the operators. The aim is simply to make later definitions unambiguous: everything in the theory is expressed using these few constructs, and the table below summarises them for reference.

\begin{table}[h!]
	\centering
	\caption{Core notation in Hypernetwork Theory.}
	\label{tab:notation}
	\begin{tabular}{@{}p{3cm}p{12cm}@{}}
		\toprule
		\textbf{Symbol} & \textbf{Meaning} \\ \midrule
		$H$ & A hypernetwork: a finite collection of hypersimplices. \\
		$v$ & Vertex (atomic identifier). \\
		$\tilde{v}$ & Anti-vertex (explicit exclusion, introduced by the modeller or by operators). \\
		$\sigma$ & Simplex: ordered tuple of vertices or hypersimplices. \\
		$\varsigma$ & Hypersimplex bound to a relation $R$. \\
		$\alpha/\beta$ & Aggregation type: conjunctive or disjunctive. \\
		$R$ & Relation symbol fixing arity and roles. \\
		$B(H)$ & Set of boundaries defined in $H$. \\
		$B(H,b) = \pi_b(H)$ & Projection of $H$ by boundary $b$, yielding a valid sub-hypernetwork. \\
		$\sqcup, \sqcap, /, \ominus, \pi$ & Merge, meet, difference, prune, split operators. \\
		\bottomrule
	\end{tabular}
\end{table}

\subsection{Relations and the Role of $R$}
Every hypersimplex is bound to a relation symbol $R$ that fixes arity, maps participants to ordered roles, and anchors semantic interpretation. A micro-example helps clarify:
\[
\varsigma^\alpha_{\text{PatientAdmission}} = \langle \textit{Patient, Ward, Time} ; R_{\text{admittedTo}} \rangle,
\]
where $R_{\text{admittedTo}}$ fixes a ternary relation with explicit participant–role alignment.

Relation symbols may be \textbf{explicit}, \textbf{implicit}, or \textbf{anonymous}, depending on whether meaning is declared or inferred from structure. Combined with $\alpha/\beta$ typing, $R$ makes the hypersimplex the fundamental semantic unit of HT.

\subsection{Boundaries}
A boundary is an element $b \in B(H)$ with projection function
\[
B(H,b) = \pi_b(H).
\]
Consider the following micro-example:
\[
\varsigma^\alpha_{\text{TriageTent}} = \langle \textit{TentFrame}, \textit{Supplies}, \textit{Medic1}; R_{\text{facility}} ; b_{\text{Logistics}} \rangle.
\]
A non-percolating boundary applies only to the tagged element; a percolating boundary propagates to descendants. Boundaries act as scoping devices rather than containers: identity of vertices and hypersimplices is preserved globally. Projection is semantics-preserving and yields a closed sub-hypernetwork aligning with the structural projection $\pi_b(H)$.

\subsection{Ordering}
Although a hypernetwork is written as a collection, algorithms execute sequentially. Ordering matters:
\begin{itemize}[leftmargin=*]
	\item within simplices, where order determines role mapping under $R$;
	\item across hypersimplices, where insertion order may influence algorithmic outcomes.
\end{itemize}
A minimal ordering-sensitive example:
\[
\langle \textit{Patient, Clinician, Time}; R_{\text{visit}} \rangle \neq \langle \textit{Clinician, Patient, Time} ; R_{\text{visit}} \rangle
\]
because role alignment differs.

Deterministic decision tables ensure reproducible results for any given ordering.

\subsection{Open-world and Closed-world Assumptions}
The \textit{Closed World Assumption} (CWA)~\cite{Keet2013b} treats absence as negation, whereas the \textit{Open World Assumption} (OWA)~\cite{Keet2013a} treats absence as unknown. HT blends these: openness applies to scope, permitting extension; closure applies to rules, ensuring every construct is well-formed, boundaries respected, and exclusions explicit.

A simple micro-example illustrates OWA. If only
\[
\varsigma^\alpha_{\text{Car}} = \langle \textit{Body, Engine, Gears} ; R_{\text{hasParts}} \rangle
\]
is present, nothing in the model implies that \emph{Wheels} do not exist—only that they have not yet been stated. This reconciliation makes hypernetworks both extensible and mechanisable, providing the semantic basis for deterministic structural composition.

\section{Core Axioms}
\label{sec:axioms}

Hypernetwork Theory is grounded in a small set of design principles:  
(i) every entity has a globally unique identity;  
(ii) exclusions may be made explicit;  
(iii) all $n$-ary relations are typed as conjunctive ($\alpha$) or disjunctive ($\beta$);  
(iv) every relation is bound to a symbol $R$ fixing arity and ordered roles; and  
(v) boundaries act as scoping devices without altering identity.  
These informal principles are formalised below as axioms A1–A5.  
Micro-examples are included to anchor each axiom in a minimal, concrete construct.

\paragraph{Axiom A1: Vertex and Hypersimplex Uniqueness.}
Every vertex and hypersimplex carries a globally unique identifier. For instance, the two hypersimplices
\[
\varsigma^\alpha_{\text{WheelAssembly}} = \langle \textit{Wheel, Axle}; R_{\text{config}} \rangle,
\qquad
\varsigma^\alpha_{\text{WheelAssembly}'} = \langle \textit{Wheel, Axle} ; R_{\text{config}} \rangle
\]
must be treated as distinct unless declared identical. Identity is preserved under projection and all structural operators, ensuring consistent composition and avoiding ambiguity.

\paragraph{Axiom A2: Explicit and Semantically Induced Exclusion.}

A modeller may introduce an anti-vertex $\medtilde{v}$ to mark explicit exclusion. For example,
\[
\varsigma^\alpha_{\text{Wheels}}=\langle\textit{Wheel}_1, \textit{Wheel}_2, \textit{Wheel}_3, \textit{Wheel}_4, \medtilde{\textit{SpareWheel}}; R_{\text{carWheels}}\rangle,
\]
expresses intentional absence within a configuration.

Anti-vertices may also be introduced by HT operators only when explicitly required by their decision tables (e.g., by prune). Boundaries and projection never generate anti-vertices, and they preserve all modeller-supplied ones unchanged.

\paragraph{Axiom A3: Aggregation Typing.}
Every hypersimplex is typed as $\alpha$ or $\beta$. Micro-examples:
\[
\varsigma^\alpha_{\text{Car}} = \langle \textit{Body, Engine, Wheels}; R_{\text{hasParts}} \rangle
\qquad
\varsigma^\beta_{\text{VehicleType}} = \{ \textit{Car, Van, Truck} ; R_{\text{isA}} \}.
\]
Typed $n$-ary structure avoids the reconstruction steps required in untyped hypergraphs or triple-based ontologies and guarantees deterministic behaviour.

\paragraph{Axiom A4: Relation Binding.}
Each hypersimplex is bound to a relation symbol $R$ fixing arity and ordered roles. For example:
\[
\varsigma^\alpha_{\text{Admission}} = \langle \textit{Patient, Ward, Time} ; R_{\text{admittedTo}} \rangle,
\]
where $R_{\text{admittedTo}}$ specifies a ternary relation with fixed role alignment. Binding ensures that semantics and participant–role mapping are explicit and mechanisable.

\paragraph{Axiom A5: Boundary Scoping.}
For a hypernetwork $H$, each boundary $b \in B(H)$ delimits scope without altering identity. A simple example:
\[
\varsigma^\alpha_{\text{TriageTent}} = \langle \textit{TentFrame}, \textit{Supplies}, \textit{Medic1}; R_{\text{facility}} ; b_{\text{Logistics}} \rangle.
\]
Projection $B(H,b)=\pi_b(H)$ preserves the identifiers of all included elements and yields a valid sub-hypernetwork. Boundaries may be non-percolating (applied locally) or percolating (propagating to descendants), supporting structured visibility and controlled recomposition.

\medskip

Together, A1–A5 ensure that hypernetworks are well-formed, semantically explicit, and amenable to deterministic manipulation. They reconcile openness with rigour: models remain extensible under the Open World Assumption, yet every construct and operation must satisfy these rules of validity. Appendix~A analyses their consequences, including closure under projection and the behaviour of modeller-supplied exclusion.

\section{Structural Properties and Composition}
\label{sec:structure-and-composition}

Hypernetwork Theory (HT) treats a hypernetwork (\textit{Hn}) not as a loose collection of tuples but as a coherent, semantically explicit structure governed by axioms~A1--A5. Once identity, exclusion, typing, relation binding, and boundary scoping are enforced, the \textit{Hn} exhibits a set of structural properties that make principled composition possible. The micro-examples below anchor each property in minimal concrete structure.

\subsection{Structural Properties of a Hypernetwork}

\paragraph{Closure.}
Each operator is closed over the class of hypernetworks: valid inputs yield valid outputs. Consider a simple name-lift example:
\[
\textit{WheelAssembly}, \qquad
\varsigma^\alpha_{\text{WheelAssembly}}
= \langle \textit{Rim, Hub, Tyre} ; R_{\text{wheelAssembly}} \rangle.
\]
If a bare vertex \textit{WheelAssembly} is present in $H$ and a hypersimplex with the same identifier is later introduced, the hypersimplex replaces the vertex so that no dangling or ill-typed structure arises. Under the Open World Assumption, additional hypersimplices may be added freely, but every intermediate state must remain well-formed under A1--A5.

\paragraph{Consistency.}
Uniqueness, type correctness, relation binding, ordering, and boundary scoping keep the hypernetwork coherent as it evolves. For instance:
\[
\varsigma^\beta_{\text{VehicleType}} = \{ \textit{Car, Van, Truck} ; R_{\text{isA}} \}
\]
cannot be merged with a hypersimplex using the same identifier but a different aggregation type. Anti-vertices such as $\medtilde{\textit{SpareWheel}}$ allow explicit exclusion without violating well-formedness.

\paragraph{Explicit Semantics.}
Semantics are embedded directly in each construct. In
\[
\varsigma^\alpha_{\text{Admission}} = \langle \textit{Patient, Ward, Time} ; R_{\text{admittedTo}} \rangle,
\]
the ternary relation $R_{\text{admittedTo}}$ fixes arity and role alignment, removing any need for semantic reconstruction. Boundaries add scoped semantics without altering identity.

\paragraph{Algorithmic Determinism.}

A hypernetwork is extensionally a set of hypersimplices and therefore unordered. For definition and storage,
\[
H = \{ \varsigma_1, \varsigma_2, \varsigma_3 \},
\]
carries no intrinsic sequence.  Operationally, however, algorithms must process these hypersimplices in a fixed traversal order: without a stable sequence, the outcome of composition may vary between runs.

Each hypersimplex has its own internal role-order that must be preserved, e.g.
\[
\langle \textit{Patient, Clinician, Time}; R_{\text{visit}} \rangle \neq \langle \textit{Clinician, Patient, Time} ; R_{\text{visit}} \rangle
\]

Determinism therefore requires two forms of ordering: a fixed traversal order for the hypernetwork, and preservation of the internal role-order of each hypersimplex.

\paragraph{Structural Integrity.}
Closure, consistency, explicit semantics, and determinism together ensure that a hypernetwork remains stable under manipulation. For example, a boundary projection of:
\[
\varsigma^\alpha_{\text{TriageTent}} = \langle \textit{TentFrame}, \textit{Supplies}, \textit{Medic1}; R_{\text{facility}} ; b_{\text{Logistics}} \rangle.
\]
preserves identity, aggregation semantics, and role alignment. This integrity supports reliable comparison, merge, decomposition, and projection.

\subsection{Structural Operators}

These properties make semantics-preserving composition possible. HT provides five core operators:

\begin{itemize}[leftmargin=*]
	\item \textbf{Merge ($\sqcup$)} integrates two hypernetworks when relation symbols, aggregation types, and roles align.  
	\emph{Micro-example:}
	\[
	\varsigma^\alpha_{\text{Car}} \sqcup 
	\varsigma^\alpha_{\text{WheelAssembly}}
	\]
	yields their union if types, roles, and identifiers do not conflict.
	
	\item \textbf{Meet ($\sqcap$)} extracts common structure, respecting $\alpha/\beta$ semantics.  
	\emph{Example:}
	\[
	\{\textit{Car, Van} ; R_{\text{isA}}\} \sqcap \{\textit{Car, Truck} ; R_{\text{isA}}\} = \{\textit{Car} ; R_{\text{isA}}\}.
	\]
	
	\item \textbf{Difference ($/$)} isolates structure in one hypernetwork but not another.  
	\emph{Example:}
	\[
	\{Car, Van ; R_{\text{isA}}\} / \{Car ; R_{\text{isA}}\} = \{Van ; R_{\text{isA}}\}.
	\]
	
	\item \textbf{Prune ($\ominus$)} removes selected structure while preserving closure. Each chosen participant is replaced everywhere by its anti-vertex; any hypersimplex that becomes invalid or fully pruned is deleted; and any vertex or hypersimplex rendered unreachable is removed. Elements still referenced by surviving hypersimplices are always retained.
	
	\emph{Example:} pruning \texttt{Engine} from
	\[
	\langle \textit{Body}, \varsigma^\alpha_{\text{Engine}}, \textit{Wheels}; R_{\text{hasParts}} \rangle
	\]
	yields
	\[
	\langle \textit{Body}, \tilde{\textit{Engine}}, \textit{Wheels}; R_{\text{hasParts}} \rangle .
	\]
	The rewritten hypersimplex is deleted only if $R_{\text{hasParts}}$ forbids an excluded role, and $\varsigma^\alpha_{\text{Engine}}$ is removed only if no remaining hypersimplex in $H$ refers to it.
	
	\item \textbf{Split ($\pi$)} projects a sub-hypernetwork via boundaries or vertex-based closure.  
	\emph{Example:}
	\[
	B(H,b_{\text{Logistics}}) = \pi_{b_{\text{Logistics}}}(H)
	\]
	selects only logistics-scoped elements, such as \texttt{TriageTent}, while preserving identifiers.
\end{itemize}

These operators refine basic set-theoretic operations by embedding semantic constraints. Merge, meet, and difference embed relational meaning, boundary behaviour, and aggregation semantics; prune supports controlled editing; split generalises projection through boundary-aware or role-closed selection. Each operator follows an explicit decision table, ensuring deterministic, semantics-preserving behaviour.

\section{Axiomatic Conditions for Validity}
\label{sec:axioms-validity}

For composition to be semantics-preserving, the operators of Section~\ref{sec:structure-and-composition} must act under conditions that guarantee closure, consistency, and determinism. These are not new axioms; they are \emph{specialisations} of the core principles established earlier. Whereas A1--A5 ensure that a hypernetwork is valid in isolation, C1--C7 extend this to \emph{compositional} validity. Micro-examples are included to show how each condition manifests in practice.

\paragraph{Core axioms (A1--A5).} Validity requires globally unique identifiers (A1); optional explicit anti-vertices for exclusion (A2); mandatory $\alpha/\beta$ typing (A3); binding of each hypersimplex to a relation $R$ with ordered roles (A4); and boundary scoping with defined percolation behaviour (A5).

A minimal example satisfying all five is: 
\[
\varsigma^\alpha_{Admission} = \langle Patient, Ward, Time ; R_{\text{admittedTo}} ; b_{\text{Medical}} \rangle, \qquad \medtilde{\textit{MissingTime}} \in H,
\] 
where $R_{\text{admittedTo}}$ fixes arity and role ordering, $b_{\text{Medical}}$ scopes visibility, and $\medtilde{\textit{MissingTime}}$ is a modeller-supplied anti-vertex marking an explicitly excluded value. These principles make every hypernetwork semantically explicit, structurally closed, and extendable under the Open World Assumption.

\paragraph{Operator conditions (C1--C7).}
Specialising A1--A5 for composition yields the following constraints. Each is illustrated with a minimal construct showing when the condition succeeds or blocks composition.

\begin{itemize}[leftmargin=*]
	\item[C1.] \textbf{$R$-compatibility.}  
	Only hypersimplices with matching relation symbols may compose; arity and ordered roles must align.  
	\[
	\langle A, B ; R_{\text{pair}} \rangle \quad\text{can merge with}\quad \langle A, B ; R_{\text{pair}} \rangle,
	\]
	but not with
	\[
	\langle A, B, C ; R_{\text{triple}} \rangle.
	\]
	
	\item[C2.] \textbf{Aggregation preservation.}  
	$\alpha/\beta$ typing is immutable.  
	\[
	\langle \textit{Body, Engine, Wheels} ; R_{\text{hasParts}} \rangle \quad\text{cannot merge with}\quad \{ \textit{Body, Engine, Wheels}; R_{\text{hasParts}} \}.
	\]
	
	\item[C3.] \textbf{Explicit exclusion discipline.}  
	Anti-vertices may be preserved or introduced by operators, but only under explicit rules; no operator introduces them arbitrarily or relies on implicit exclusion.
	\[
	\medtilde{\textit{SpareWheel}} \; \in H_1,\quad \medtilde{\textit{SpareWheel}} \; \notin H_2
	\]
	does not force $H_2$ to include it; operators simply preserve it when encountered.
	
	\item[C4.] \textbf{Boundary preservation.}  
	All boundary tags must be preserved through composition.  
	\[
	\varsigma^\alpha_{\text{TriageTent}} = \langle \textit{TentFrame}, \textit{Supplies}, \textit{Medic1}; R_{\text{facility}} ; b_{\text{Logistics}} \rangle.
	\]
	remains tagged with $b_{\text{Logistics}}$ under merge, meet, difference, prune, and split. Projection $B(H,b)=\pi_b(H)$ ensures scope-respecting submodels.
	
	\item[C5.] \textbf{OWA with rule-based closure.}  
	Operators must return well-formed hypernetworks. If pruning removes a participant:
	\[
	\langle \textit{Body, Engine, Wheels}; R_{\text{hasParts}} \rangle \ominus \{ \textit{Engine, Body} \}, \qquad \langle \medtilde{\textit{Body}}, \medtilde{\textit{Engine}}, \textit{Wheels}; R_{\text{hasParts}} \rangle,
	\]
	the hypersimplex is removed to maintain closure; no invalid partial structure is allowed.
	
	\item[C6.] \textbf{Ordering determinism.}  
	Role order is fixed, and insertion order yields deterministic outcomes.  
	\[
	\langle \textit{Patient, Clinician, Time}; R_{\text{visit}} \rangle \neq \langle \textit{Clinician, Patient, Time} ; R_{\text{visit}} \rangle
	\]
	and a deterministic merge sequence ensures reproducibility for any fixed ordering.
	
	\item[C7.] \textbf{Algorithmic soundness.}  
	Procedures must enforce A1--A5 throughout. The reference implementation applies uniqueness checks, role alignment, type checks, and boundary preservation at each step, guaranteeing that every state of the hypernetwork remains valid.
\end{itemize}

\medskip

Together, A1--A5 and C1--C7 define the full validity regime for structural composition. Operators take hypernetworks as operands and return new hypernetworks while preserving identity, typing, relation binding, and scope. This contrasts with approaches where composition is syntactic, tool-dependent, or defined only informally. Here, composition is governed by explicit axioms and executable conditions, making the algebra principled and mechanisable. Section~\ref{sec:algebraic} analyses the algebraic consequences, and Section~\ref{sec:mechanisation} shows how these conditions are enforced algorithmically.

\section{Algebraic Properties}
\label{sec:algebraic}

The structural operators of Section~\ref{sec:structure-and-composition} form a coherent \emph{algebra} of composition for hypernetworks. Earlier treatments considered merge, meet, difference, and split individually; here these results are unified, extended to include prune, and grounded collectively in A1--A5 and the operator conditions C1--C7. The micro-examples below illustrate each algebraic behaviour in minimal form.

\paragraph{Closure.}
Each operator is closed over the class of hypernetworks: valid inputs yield valid outputs. Consider a simple name-lift example:
\[
\textit{WheelAssembly}, \qquad
\varsigma^\alpha_{\text{WheelAssembly}} = \langle \textit{Rim, Hub, Tyre} ; R_{\text{wheelAssembly}} \rangle.
\]
If a bare vertex \textit{WheelAssembly} is present in $H$ and a hypersimplex with the same identifier is later introduced, the hypersimplex replaces the vertex so that no dangling or ill-typed structure arises. Under the Open World Assumption, further hypersimplices may be added freely, but every intermediate result must remain well-formed under A1--A5.

\paragraph{Non-associativity and non-commutativity.}
The operators behave as structured, semantics-aware transformations. For merge:
\[
(H_1 \sqcup H_2) \sqcup H_3 \quad\neq\quad H_1 \sqcup (H_2 \sqcup H_3)
\]
whenever ordering exposes different yet valid alignments. Similarly, meet and difference behave differently under reordering:
\[
H_1 / H_2 \neq H_2 / H_1, \qquad (H_1 \sqcap H_2) \sqcap H_3 \neq H_1 \sqcap (H_2 \sqcap H_3).
\]
These behaviours reflect legitimate alternatives arising from ordering and scope, not defects.

\paragraph{Idempotence.}
Merge, meet, prune, nd split are idempotent:
\[
H \sqcup H = H, \qquad H \sqcap H = H, \qquad (H \ominus S) \ominus S = H \ominus S \qquad pi_b(\pi_b(H)) = \pi_b(H).
\]
Difference is not idempotent:
\[
H / H = \emptyset, \quad (H / H_1) / H_1 = \emptyset,
\]
Percolation affects tag assignment when hypersimplices are introduced or rewritten, but split itself never propagates tags; it is a pure projection over the tags already present.

\paragraph{Monotonicity.}
Let $\sqsubseteq$ denote the sub-hypernetwork relation. Merge, meet, and split preserve monotonic growth:
\[
H_1 \sqsubseteq H_1 \sqcup H_2, \qquad H_1 \sqcap H_2 \sqsubseteq H_1, \qquad \pi_b(H) \sqsubseteq H.
\]
Difference is monotone in its first argument and antitone in its second:
\[
H_1 \sqsubseteq H_1' \Rightarrow (H_1 / H_2) \sqsubseteq (H_1' / H_2), \qquad H_2 \sqsubseteq H_2' \Rightarrow (H_1 / H_2') \sqsubseteq (H_1 / H_2).
\]
Prune is monotone in its deletion set:
\[
S \subseteq S' \Rightarrow (H \ominus S') \sqsubseteq (H \ominus S).
\]

\paragraph{Identity elements.}
The empty hypernetwork is the identity for merge:
\[
H \sqcup \emptyset = H,
\]
and the zero for meet:
\[
H \sqcap \emptyset = \emptyset.
\]
Difference and prune satisfy:
\[
H / \emptyset = H, \qquad H / H = \emptyset, \qquad H \ominus \emptyset = H.
\]
Split preserves identity under the universal criterion:
\[
\pi_{\top}(H) = H.
\]

\paragraph{Semantics preservation.}
All operators preserve $\alpha/\beta$ typing, relation binding, vertex identity, and boundary tags. For example:
\[
\langle \textit{TentFrame, Supplies} ; R_{\text{facility}} ; b_{\text{Logistics}} \rangle
\]
retains $b_{\text{Logistics}}$ under merge, meet, difference, prune, and split. Anti-vertices, if present, are preserved, and may be introduced by operators only under the explicit exclusion rules of A2 and C3. Each operator maintains structural validity: meet retains overlapping participants, prune removes dependencies but not identity, and split yields closed sub-hypernetworks with preserved scope.

\paragraph{Determinism.}
Determinism means that once the modeller fixes an ordering of inputs, the operator produces one unique result for that ordering. This is easiest to see in a small micro-example where a bare vertex later appears as a typed hypersimplex. Consider the ordered input
\[
\textit{Patient} \quad\text{followed by}\quad 
\langle \textit{Patient, Clinician, Time}; R_{\text{visit}} \rangle.
\]
Here the vertex \textit{Patient} is inserted first. When the hypersimplex arrives, it deterministically replaces the vertex, because the hypersimplex is the more specific construct tied to a relation symbol and ordered roles. Reversing the order,
\[
\langle \textit{Patient, Clinician, Time}; R_{\text{visit}} \rangle 
\quad\text{followed by}\quad 
\textit{Patient},
\]
still gives a unique result: the second step is a deterministic ``ignore'', because the hypersimplex already defines that identifier and a bare vertex adds no new information. These two orders produce two different but equally valid hypernetworks, yet each order has exactly one outcome. That is the sense in which composition is deterministic: the modeller may choose different input orders, but once the order is fixed, the result is fixed.

\medskip

In summary, hypernetwork composition defines a mechanisable algebra: closed, semantics-preserving, deterministic under fixed order, and sensitive to ordering and scope where this reflects underlying multilevel semantics. Incorporating prune, clarifying split, and analysing the operators as a unified algebra completes the progression from conceptual foundations to a fully executable structural framework. These results form the algebraic base on which the boundary calculus builds.

\section{Mechanisation of Structural Composition}
\label{sec:mechanisation}

The algebra of Section~\ref{sec:structure-and-composition} attains full force when expressed as executable procedures. Mechanisation---including \texttt{prune} and the clarified \texttt{split}---demonstrates that the operators of Hypernetwork Theory can be realised as deterministic algorithms governed by explicit decision tables. No heuristics or tool conventions are required: operators act reproducibly on hypernetworks, preserving closure and semantics under A1--A5 and the operator conditions C1--C7. This contrasts with UML, SysML, and MDE, where behaviour is often tool-specific and not semantically grounded~\cite{delligatti2013sysml,eriksson2003uml,omg2014mda}.

\subsection{Reference Implementation in Hora}

Mechanisation is validated in the reference implementation, which provides parser, schema, and operator support for the refined notation and axioms. Boundaries are represented as elements $b \in B(H)$ with projections $B(H,b)=\pi_b(H)$, ensuring that scope-based and structure-based composition coincide.

\paragraph{Micro-example.}
Given a hypernetwork
\[
H = \{
\varsigma^{\alpha}_{\text{IncidentA}},
\varsigma^{\alpha}_{\text{TeamBlue}},
\varsigma^{\alpha}_{\text{TriageTent}}
\},
\]
a direct call to \texttt{split(H, b\_Medical)} in Hora yields:
\[
\pi_{b_{\text{Medical}}}(H)
=
\{
\varsigma^{\alpha}_{\text{IncidentA}},
\varsigma^{\alpha}_{\text{TeamBlue}},
\varsigma^{\alpha}_{\text{TriageTent}}
\},
\]
because all three hypersimplices carry $b_{\text{Medical}}$ or inherit it through percolation. Identity is preserved automatically; no reconstruction is required.

Operators are encoded as inline procedures following the same semantics as their formal definitions. The mechanisation of \texttt{prune} and the clarified \texttt{split} completes the algebra of composition. Scoped boundary behaviour---percolation, commutation, and projection---is handled uniformly and aligns with identity-preserving locality found in abstract frameworks such as sheaf theory~\cite{curry2014sheaves,maclane1992sheaves}.

\subsection{General Algorithmic Framework}

All operator algorithms follow a common pattern:
\begin{enumerate}[leftmargin=*]
	\item validate vertices---enforce A1 and respect explicit anti-vertices introduced under A2;
	\item validate hypersimplices---check relation compatibility (C1), aggregation type (A3, C2), and preserve boundary tags (A5, C4);
	\item apply the operator’s decision table elementwise in sequence (C6);
	\item reapply A1--A5 and C1--C7 after each insertion or removal to guarantee closure (C7).
\end{enumerate}

\paragraph{Micro-example.}
For ordered input
\[
[\varsigma^{\alpha}_{\text{TeamBlue}},\varsigma^{\alpha}_{\text{IncidentA}}],
\]
and operator \texttt{merge}, the algorithm:
\[
H_0 = \emptyset,
H_1 = \{ \varsigma^{\alpha}_{\text{TeamBlue}}\},
H_2 = H_1  \sqcup \{ \varsigma^{\alpha}_{\text{IncidentA}}\}
\]
yields a deterministic final hypernetwork regardless of tool implementation, because every step applies A1--A5 and C1--C7 explicitly. Reordering the input gives a different but still valid $H_2$, confirming controlled nondeterminism rather than tool-specific behaviour.

\subsection{Insert Algorithm}

The \textsc{Insert} routine realises the merge decision table and provides the skeleton for all operator implementations. During construction, \textsc{Insert} adds elements one at a time to the empty hypernetwork; merge generalises this to combining full structures. Name clashes, aggregation-type mismatches, and incompatible relation symbols are rejected; compatible hypersimplices are unified; otherwise new elements are created. Bidirectional containment is maintained automatically: whenever a hypersimplex lists participants in its \texttt{simplex}, reciprocal \texttt{partOf} links are added, giving closure by construction. The full decision tables and procedural specification are given in Appendix~B.

\paragraph{Micro-example.}
Calling \textsc{Insert} on
\[
\varsigma^{\alpha}_{\text{TeamBlue}} = \langle \textit{Commander, Medic1, Medic2} ; R_{\text{Team}} \rangle
\]
first inserts the vertices (A1), then the simplex, then the hypersimplex. A second call with
\[
\langle \textit{Commander, Medic1, Medic2}; R_{\text{Team}} \rangle
\]
is ignored by idempotence, as the decision table classifies it as \emph{identical hypersimplex}. A call with
\[
\langle \textit{Commander, Medic1, Medic2, Extra}; R_{\text{Team}} \rangle
\]
is rejected because $R_{\text{Team}}$ fixes arity, enforcing C1 (relation compatibility) and A4 (role binding).

\subsection{Summary}

Mechanisation reduces to a simple pattern: apply the operator’s decision table, enforce closure at each step, and preserve order. The algebra of Section~\ref{sec:structure-and-composition} is therefore not merely conceptual but fully executable. By grounding execution in axioms and explicit operator conditions, HT provides operators that are deterministic, semantics-preserving, and straightforward to implement. Scoped boundary behaviour follows the same pattern, ensuring that structural composition remains valid even under overlapping or heterogeneous boundaries.

\section{Worked Example: A Multilevel Emergency Response System}
\label{sec:examples}

This flagship example demonstrates how typed $n$-ary relations, boundaries, identity preservation, and semantics-preserving operators support structures that conventional formalisms cannot represent without reconstruction, hierarchy-forcing, or loss of semantics. Emergency response systems contain overlapping roles, multilevel dependencies, and cross-agency coordination. HT captures these directly, without tool-dependent semantics.

\subsection{System Structure}

An incident is modelled as a typed $n$-ary $\alpha$-aggregation:
\[
\begin{aligned}
	\varsigma^{\alpha}_{\text{IncidentA}} = \langle & \textit{Commander, UnitBlue, UnitRed, Casualties, HospitalX, Status}; \\
	& R_{\text{Incident}} ; b_{\text{Ops}}, b_{\text{Medical}}, b_{\text{IncidentA}} \rangle.
\end{aligned}
\]

The same \emph{Commander} participates simultaneously across levels and roles:
\[
\begin{aligned}
	\varsigma^{\alpha}_{\text{TeamBlue}}
	&= \langle \textit{Commander, Medic1, Medic2}; R_{\text{Team}} ; b_{\text{Ops}}, b_{\text{Medical}} \rangle, \\[2mm]
	\varsigma^{\beta}_{\text{Rank}}
	&= \{ \textit{Commander, Deputy, DutyOfficer}; R_{\text{Rank}} ; b_{\text{Ops}} \}.
\end{aligned}
\]

This shows $\alpha$ (part--whole) and $\beta$ (taxonomic) semantics coexisting without collapsing them into binary associations or role conventions.

A triage tent is simultaneously equipment and a clinical micro-facility:
\[
\varsigma^{\alpha}_{\text{TriageTent}} = \langle \textit{TentFrame, Supplies, Medic1}; R_{\text{Facility}} ; b_{\text{Logistics}}, b_{\text{Medical}} \rangle.
\]
Its identity is preserved across logistics and clinical subsystems, illustrating heterarchy.

\subsection{Boundaries and Scoped Projection}

Operational subsystems are delimited by boundaries:
\[
b_{\text{Ops}}, \qquad b_{\text{Medical}}, \qquad b_{\text{Logistics}}, \qquad b_{\text{IncidentA}}.
\]

A percolating medical projection extracts clinically relevant structure:
\[
B(H, b_{\text{Medical}}) = \pi_{b_{\text{Medical}}}(H).
\]

Explicitly:
\[
\begin{aligned}
	\pi_{b_{\text{Medical}}}(H)
	= \{ & \textit{Medic1}, \textit{Medic2}, \textit{Casualties}, \textit{HospitalX}, \\
	& \varsigma^{\alpha}_{\text{TriageTent}},
	\varsigma^{\alpha}_{\text{TeamBlue}},
	\varsigma^{\alpha}_{\text{IncidentA}} \}.
\end{aligned}
\]

Identity is preserved: \emph{Commander} remains the same structural element across all scopes.

\subsection{Structural Composition}

Two subsystems are developed independently:
\[
H_{\text{Ops}},\qquad H_{\text{Clinical}}.
\]

\paragraph{Merge.}
The operator recombines them into a coherent whole:
\[
H_{\text{Merged}} = H_{\text{Ops}}  \sqcup  H_{\text{Clinical}}.
\]

Concretely:
\[
\begin{aligned}
	H_{\text{Ops}} &= 
	\{
	\varsigma^{\alpha}_{\text{IncidentA}},
	\varsigma^{\alpha}_{\text{TeamBlue}}
	\},\\[1mm]
	H_{\text{Clinical}} &= 
	\{
	\varsigma^{\alpha}_{\text{IncidentA}},
	\varsigma^{\alpha}_{\text{TriageTent}},
	\varsigma^{\beta}_{\text{PatientClass}}
	\},\\[1mm]
	H_{\text{Merged}} &=
	\{
	\varsigma^{\alpha}_{\text{IncidentA}},
	\varsigma^{\alpha}_{\text{TeamBlue}},
	\varsigma^{\alpha}_{\text{TriageTent}},
	\varsigma^{\beta}_{\text{PatientClass}}
	\}.
\end{aligned}
\]

\paragraph{Meet.}
\[
H_{\text{Common}} = H_{\text{Ops}}  \sqcap  H_{\text{Clinical}},
\]
revealing shared elements such as \emph{Commander}, \emph{Medic1}, and \emph{HospitalX}.

\paragraph{Difference.}
\[
H_{\text{OpsOnly}} = H_{\text{Ops}} / H_{\text{Clinical}}.
\]

\paragraph{Prune.}
Standing down \texttt{UnitRed} yields:
\[
H' = H_{\text{Merged}} \,\ominus\, \{ \textit{UnitRed} \}.
\]
giving an anti-vertex under A2:
\[
\begin{aligned}
	\varsigma^{\alpha}_{\text{IncidentA}} =
	\langle
	&\textit{Commander, UnitBlue},\, \medtilde{\textit{UnitRed}},\, \textit{Casualties},\, \textit{HospitalX},\, \textit{Status} ;\\
	&R_{\text{Incident}} ; b_{\text{Ops}},\, b_{\text{Medical}},\, b_{\text{IncidentA}}
	\rangle.
\end{aligned}
\]
Closure is preserved; no dangling structure remains.

\subsection{Significance}

\paragraph{Multilevel and heterarchical structure.}
Typed hypersimplices represent part--whole, taxonomic, and cross-level dependencies directly, without forcing hierarchy or reproducing semantics via inference.

\paragraph{Scoped, identity-preserving views.}
Boundaries generate subsystem projections in which identity is preserved, enabling analysis and recomposition across overlapping contexts.

\paragraph{Semantics-preserving composition.}
Merge, meet, difference, prune, and split operate on typed $n$-ary structure with guaranteed closure and determinism, unlike graph-based, ontology-based, or diagram-centric approaches.

\subsection{Comparison with Established Formalisms}

\begin{center}
	\begin{tabular}{lcccc}
		\toprule
		\textbf{Property} & \textbf{UML/SysML} & \textbf{OWL/RDF} & \textbf{Hypergraphs} & \textbf{HT} \\
		\midrule
		Typed $n$-ary relations & No & No & No & Yes \\
		Identity across boundaries & No & No & No & Yes \\
		Explicit $\alpha/\beta$ semantics & No & No & No & Yes \\
		Scoped projection & No & No & No & Yes \\
		Cross-level structure & Limited & No & No & Yes \\
		\bottomrule
	\end{tabular}
\end{center}

This worked example shows why HT functions as a general structural kernel: explicit typed structure; principled scoping; identity stability; and deterministic, semantics-preserving composition across heterogeneous subsystems.

\section{Discussion}
\label{sec:discussion}

This paper has established Hypernetwork Theory (HT) as a semantics-preserving and mechanisable framework for multilevel $n$-ary modelling. By integrating conceptual motivation, refined notation, axioms, structural operators, and executable algorithms into a single account, the theory becomes both expressive and operational. Typed aggregation, explicit relation binding, scoped boundaries, optional exclusion, and ordering together provide a minimal core that guarantees well-formedness, identity preservation, and deterministic behaviour.

\subsection{Generality}
The constructs are domain-independent. Typed $\alpha/\beta$ aggregation, ordered role alignment, and scoped boundaries arise across engineering systems, organisational structure, socio-technical arrangements, biological networks, and formal ontologies. Comparable requirements appear in hypergraph and higher-order approaches~\cite{RN363,RN365,RN197}, ontology formalisms~\cite{RN688,RN666,RN664,brickley2004rdf}, and abstract mathematical treatments of locality and composition~\cite{curry2014sheaves,goguen1991manifesto,maclane1992sheaves,spivak2014ct}. HT therefore provides a notation-independent account of semantically explicit relational structure rather than a domain-specific language.

\subsection{Interoperability}
Because semantics are embedded directly in its primitives—rather than inferred from diagrams, tool conventions, or inference engines—HT offers a reproducible basis for interoperability. Graphs and hypergraphs capture connectivity but not typed relational meaning~\cite{RN363,RN197}; ontologies reconstruct $n$-ary relations indirectly~\cite{RN688,brickley2004rdf}; UML, SysML, and MOF-based frameworks remain diagram-centric and binary at their core~\cite{delligatti2013sysml,eriksson2003uml,siegel2014object}. HT instead makes identity, typing, role alignment, boundaries, and explicit exclusion first-class, enabling heterogeneous models to be aligned without ad hoc reconciliation. Scoped boundaries preserve subsystem integrity, and the operator algebra provides semantics-preserving integration.

\subsection{Mechanisation}
Mandatory typing, ordered roles, identity preservation, and boundary scoping make mechanisation direct. Operators are defined by explicit decision tables, and algorithms enforce A1--A5 and C1--C7 at each step, ensuring closure, determinism, and reproducibility. Merge, meet, difference, prune, and split can be executed without heuristics, and boundary-based projection $B(H,b)=\pi_b(H)$ aligns scope with structural closure. Mechanisation is therefore a property of the core constructs rather than an external extension.

\subsection{Limitations}
Several restrictions remain intentional. Composition presumes matching names for relations and vertices; semantic alignment across vocabularies lies outside the framework. Ordering across hypersimplices can yield scope-sensitive variants, reflecting meaningful alternatives rather than inconsistency. Boundaries, in their minimal form, support scope but not the full locality structures considered in abstract settings~\cite{curry2014sheaves}. Anti-vertices remain optional, are introduced explicitly (by the modeller or by operators when required), and do not participate directly in closure. These constraints keep the kernel precise and tractable.

\subsection{Implications}
By reconciling openness with rigour—openness at the level of scope, closure at the level of rules—HT avoids the fragility of syntactic integration and the indeterminacy of inference-based formalisms. It provides a principled basis for decomposition and recomposition, preserves semantics under editing and projection, and supports subsystem modularity across heterogeneous structures. The axioms and operators therefore supply a coherent algebra for structural composition.

\subsection{Refinements}
The account sharpens earlier relational and higher-order modelling approaches by making aggregation typing, role ordering, scoped boundaries, and identity preservation explicit. These refinements support deterministic structural operators, identity-preserving projection, and coherent subsystem views over a single immutable backcloth. They also provide a structural basis for developing constraints, dynamics, and interpretive layers without modifying the core formalism.

\medskip
In sum, HT offers a unified, mechanisable, and semantically explicit foundation for $n$-ary and multilevel modelling. The structural kernel developed here supports typed structure, deterministic composition, scoped behaviour, and principled extension towards constraints, dynamics, and interpretation.

\section{Related Work}
\label{sec:related}

Work on formal modelling spans graphs and hypergraphs, ontologies, UML and MOF-based languages, Model-Driven Engineering (MDE), multi-level modelling, architectural description languages (ADLs), and abstract mathematical frameworks such as category theory and sheaf theory. Each provides valuable insights, but none offers typed $n$-ary relations, multilevel expressiveness, scoped boundaries, and mechanisable composition within a single coherent framework.

Graphs and hypergraphs capture connectivity and higher-order structure, including multi-layer and multiplex variants~\cite{RN363,RN385,RN362,RN361}. These approaches remain analytically powerful but treat relations as untyped sets: aggregation semantics, role alignment, and part--whole meaning remain external. The term \emph{hypernetwork} also appears in ecological and biological modelling~\cite{joslyn2020hypernetwork}, but this differs fundamentally from HT, which is grounded in typed $n$-ary relations with explicit semantics and ordering.

Relational and logical frameworks---including first-order logic, relational databases, and Description Logics~\cite{RN666,RN664}---treat relations extensionally but do not provide typed $n$-ary constructs with ordered roles. RDF/OWL reification patterns~\cite{RN688,brickley2004rdf} reconstruct $n$-ary structure indirectly, obscuring semantics and preventing identity-preserving projection. These approaches lack algebraic composition and do not support role-sensitive or boundary-aware structure.

Object Orientation~\cite{booch2008object} conflates part--whole and taxonomic semantics, leaving composition informal and tool-dependent. UML and SysML~\cite{delligatti2013sysml,eriksson2003uml} retain binary associations and hierarchical packaging, with semantics driven by diagram conventions. MOF~\cite{siegel2014object} imposes fixed layering. MDE extends these notations with transformation pipelines (ATL, QVT~\cite{jouault2008atl,kurtev2006qvt,omg2016qvt}) and tooling such as EMF~\cite{steinberg2008emf}, but inherits binary meta-models and lacks semantics-preserving composition. Models remain diagram-centric, with behaviour tied to tool interpretation rather than explicit typed relational structure.

Multi-level modelling frameworks address cross-level structure via clabject mechanisms, meta-level induction, or potency rules. These approaches support flexible classification but do not provide typed $n$-ary relations, explicit aggregation semantics, or identity-preserving boundary projection. HT instead treats cross-level and heterarchical references as first-class and structurally coherent.

Architectural description languages (ADLs) such as ACME, AADL, and Wright offer structural primitives but rely on connectors, ports, and configurations that encode semantics implicitly~\cite{garlan2000acme, feiler2006architecture, allenformal}. They do not support typed $n$-ary relations, boundary-scoped projection, or deterministic structural operators. Composition is expressed syntactically or behaviourally, not through a semantics-preserving relational algebra.

Graph-transformation and rule-based frameworks provide structured rewriting but assume graph-based primitives and binary links. Algebraic approaches such as DPO/SPO offer precise transformation rules but do not embed typed $n$-ary relations, aggregation semantics, or identity-preserving scoping~\cite{ehrig2006fundamentals}. These frameworks complement HT but do not substitute for its typed relational foundation.

Abstract mathematical frameworks, including category theory~\cite{goguen1991manifesto,spivak2014ct} and sheaf theory~\cite{curry2014sheaves,maclane1992sheaves}, formalise composition, locality, and gluing but do not define typed $n$-ary relations, boundary assignment, identity-preserving projection, or deterministic structural editing. Systems-of-systems engineering~\cite{jamshidi2008sos,maier1998architecting} studies subsystem boundaries, but these are descriptive rather than structural and lack algebraic guarantees for reproducible recomposition.

Against this landscape, HT is distinctive in embedding semantics directly in its constructs. Typed $n$-ary relations are explicitly declared as $\alpha$ (conjunctive) or $\beta$ (taxonomic); aggregation and typing are preserved under projection; identity is maintained across all boundary-scoped submodels; and structural operators act deterministically under explicit rules. These properties position HT as a fully specified, mechanisable formalism capable of supporting multilevel, heterogeneous, and scoped systems modelling.

\section{Conclusion and Future Work}
\label{sec:conclusion}

This paper has presented Hypernetwork Theory (HT) as a single, structurally complete and mechanisable formalism. It unifies the conceptual motivation for typed $n$-ary relations, the refined notation and axioms, and a full algebra of structural composition. The result is a semantics-preserving, identity-stable, and executable framework for multilevel and heterogeneous systems modelling.

\paragraph{Contributions.}
This paper has delivered four main advances.
\begin{enumerate}[leftmargin=*]
	\item \textbf{A clarified relational foundation.} Typed $n$-ary relations, explicit role ordering, and $\alpha$/$\beta$ aggregation provide a precise semantic base in which diagrams become optional views rather than the primary artefact.
	\item \textbf{A formally grounded notation.} Vertices, simplices, hypersimplices, relation binding, boundaries, and ordering have been given exact definitions and unified under axioms~A1--A5, ensuring that every construct remains explicit and well-formed.
	\item \textbf{A complete structural operator algebra.} Merge, meet, difference, prune, and split have been defined through validity conditions~C1--C7 and deterministic decision tables, giving a semantics-preserving account of structural composition.
	\item \textbf{An executable framework.} Deterministic algorithms realising the operator algebra demonstrate that HT can be implemented directly and used as a reliable basis for construction, editing, and comparison.
\end{enumerate}

The formalisation of typed relations, ordered roles, relation binding, modeller-supplied exclusion, and boundary tagging provides a clear and explicit semantic foundation. The axioms guarantee well-formedness and closure under rules while maintaining an open-world stance at the level of scope. Ordering is made precise: simplicial order fixes participant--role alignment, and insertion order introduces controlled but mechanisable nondeterminism.

On this foundation, the structural operators provide deterministic procedures for merge, meet, difference, prune, and split, each grounded in explicit decision tables and operator conditions. These operators extend set-theoretic intuitions while preserving semantic typing, boundary tags, role alignment, and identity. The algorithms enforce A1--A5 and C1--C7 throughout, yielding a reproducible execution layer.

A distinctive contribution of the unified theory is its reconciliation of the Open World Assumption with rule-based closure: models remain extensible, yet every operation returns a valid hypernetwork. This contrasts with diagram-centric formalisms grounded in the Closed World Assumption and ontology-based approaches that permit openness without structural closure. HT embeds semantics directly in its primitives and preserves identity across all structural operations.

This structural kernel provides a stable base for future developments. A companion paper develops the full boundary calculus, including projection, percolation, overlap, and scoped composition. Beyond this, future work spans \emph{constraints} (scoped invariants over projections), \emph{dynamics} (time-indexed backcloths and structural evolution), and \emph{interpretation} (connections with external mathematical and engineering frameworks). The Hora environment offers a reproducible implementation baseline for these directions.

In summary, this paper establishes the structural foundation of Hypernetwork Theory as a coherent, mechanisable, and semantically explicit modelling framework. It provides a principled basis for multilevel structure, semantics-preserving composition, and scoped behaviour, and prepares the ground for constraints, dynamics, interpretation, and the boundary calculus as the next stages of development.
	
\printglossary

\section*{Acknowledgments}
Much of the work presented here traces its origins to the pioneering idea of Hypernetworks first introduced by my supervisor, Jeff Johnson.  His insight --- that systems might be understood not as collections of objects but as structured aggregations of relations --- laid the foundation for all subsequent developments.  Without that conceptual breakthrough, neither the refinements presented in this paper nor the wider programme of research into mechanisable, multilevel, and \nar{} modelling would have been possible.  I am deeply grateful for his guidance and support throughout the development of this work.

I would also like to thank my examiners, Professor Liz Varga and Dr Amel Bennaceur, for their thoughtful engagement with my thesis and for the constructive feedback that helped to strengthen its contribution.

The author acknowledges that early draft material for this paper was generated with the assistance of ChatGPT, based on material from the author’s doctoral thesis.  These drafts were subsequently reviewed, validated, and substantially revised by the author.  All final interpretations, arguments, and conclusions are the author’s own, and all references and citations have been independently selected and verified by the author.

\appendix

\setcounter{section}{1}
\renewcommand{\thesection}{A}

\section*{Appendix A: Proofs of Axioms and Operator Properties}
\addcontentsline{toc}{section}{Appendix A: Proofs of Axioms and Operator Properties}

The main text is self-contained. The axioms (A1--A5), operator conditions (C1--C7), and all structural behaviour of hypernetworks can be understood without reference to this appendix. For completeness, this appendix consolidates the key lemmas and proofs establishing that hypernetworks remain valid, semantics-preserving, and deterministic under all admissible modelling actions.

\medskip
\noindent\textbf{How to read this appendix.} Each lemma states a formal guarantee about hypernetworks or one of the structural operators. A brief gloss follows each proof to give intuition. Together they show that once a hypernetwork is valid, it remains valid under construction, projection, editing, and all operator applications.

\begin{lemma}[OWA and Closure in Construction]
	For any hypernetwork $H=\{\varsigma_1,\ldots,\varsigma_n\}$ constructed by successive insertion of hypersimplices under axioms~A1--A5, $H$ is simultaneously open to extension and closed under rules.
\end{lemma}

\begin{proof}
	Openness follows because construction is defined as iterative merge, and any new hypersimplex satisfying A1--A5 may be added. Closure follows because each merge step enforces uniqueness, typing, relation binding, and boundary scoping. Thus $H$ remains valid at every stage.
\end{proof}

\noindent\emph{Gloss.} New structure can always be added, but every step preserves validity.

\begin{lemma}[Closure with explicit anti-vertices (modeller- or operator-introduced under A2)]
	Let $H$ be valid. Replacing a vertex $v$ with its anti-vertex $\tilde{v}$ yields
	$H'=(H\setminus\{v\})\cup\{\tilde{v}\}$, and $H'$ is valid under A1--A5.
\end{lemma}

\begin{proof}
	By A1, $\tilde{v}$ is a distinct identifier. By A2, anti-vertices are optional and modeller-supplied. Typing, relation binding, and boundaries are unaffected; hence $H'$ is valid.
\end{proof}

\noindent\emph{Gloss.} Marking exclusion never breaks structure.

\begin{lemma}[Closure under Boundary Projection]
	Let $H$ be valid and $b\in B(H)$ a boundary. Then $B(H,b)=H_b$ is valid under A1--A5.
\end{lemma}

\begin{proof}
	All elements of $H_b$ are those tagged with $b$ (or descended via percolation). Identity, typing, and relation binding are inherited. A5 guarantees closure of the projection.
\end{proof}

\noindent\emph{Gloss.} Every boundary view is itself a valid hypernetwork.

\begin{lemma}[Ordering Determinism]
	Reordering hypersimplices in a valid hypernetwork does not affect validity or boundary projection.
\end{lemma}

\begin{proof}
	A1--A5 depend only on identifiers, typing, binding, and boundaries, not enumeration order. Projection depends solely on tags. Hence both remain invariant under reordering.
\end{proof}

\noindent\emph{Gloss.} Listing order does not affect the structure.

\begin{lemma}[Single-step Merge Preservation]
	Let $H$ and $K$ be valid. Successively inserting each hypersimplex of $K$ into $H$ under A1--A5 yields a valid hypernetwork.
\end{lemma}

\begin{proof}
	Assume $H_i$ valid. Inserting the next hypersimplex enforces A1--A5; thus $H_{i+1}$ is valid. Finite induction over the elements of $K$ completes the proof.
\end{proof}

\noindent\emph{Gloss.} Stepwise merging preserves validity.

\begin{lemma}[Iterated Merge Closure]
	Let $H_0$ and $K_1,\dots,K_m$ be valid. Define $H_{i+1}=\mathrm{merge}(H_i,K_i)$. Then all $H_i$ are valid.
\end{lemma}

\begin{proof}
	Immediate from single-step merge preservation by induction.
\end{proof}

\noindent\emph{Gloss.} Long merge chains remain valid.

\begin{corollary}[Operator-neutral Closure]
	Any operator realised via A1--A5-guarded insertions yields a valid hypernetwork; any finite left fold of such operations remains valid.
\end{corollary}

\noindent\emph{Gloss.} If constructed via the rules, validity is preserved.

\begin{lemma}[Closure of All Operators]
	For valid $H_1,H_2$ and any operator $\odot\in\{\sqcup,\sqcap,/,\ominus,\pi\}$, the result is valid under A1--A5.
\end{lemma}

\begin{proof}
	Decision tables enforce A1--A5 and C1--C7. Merge/meet unify only when compatible; difference subtracts without creating identifiers; prune removes dependent structure; split selects tagged or seed-closed elements. All outcomes respect typing, binding, ordering, and boundaries.
\end{proof}

\noindent\emph{Gloss.} All operators preserve structural validity.

\begin{lemma}[Open World with Rule-based Closure]
	HT remains open to extension while remaining closed under rules for all operator applications.
\end{lemma}

\begin{proof}
	Merge introduces new structure; all operators re-enforce A1--A5 and C1--C7 at each step. Thus the hypernetwork remains valid yet extensible.
\end{proof}

\noindent\emph{Gloss.} Models can grow but cannot be broken.

\begin{lemma}[Determinism]
	For identical ordered inputs, each operator yields a unique output.
\end{lemma}

\begin{proof}
	Decision tables fix evaluation order: vertices and anti-vertices first, then typing, binding, ordering, and boundaries. Algorithms therefore produce deterministic results.
\end{proof}

\noindent\emph{Gloss.} Same inputs give the same outputs.

\begin{lemma}[Boundary Projection Validity]
	For any $H$ and boundary $b$, the projection $B(H,b)=\pi_b(H)$ is valid and satisfies $B(H,b)\sqsubseteq H$.
\end{lemma}

\begin{proof}
	Projection selects tagged elements already present in $H$; axioms are preserved.
\end{proof}

\noindent\emph{Gloss.} Boundary projection always yields a valid sub-hypernetwork.

\begin{lemma}[Sub-hypernetwork Results]
	(a) If $H'$ is obtained from $H$ by $\ominus$ or $\pi$, then $H'\sqsubseteq H$. (b) For $\sqcap$ or $/$, $H'$ is a sub-hypernetwork only when all constructed hypersimplices already exist in the inputs.
\end{lemma}

\begin{proof}
	(a) prune and split do not create new hypersimplices. (b) meet/difference may construct intersected or subtracted hypersimplices; these may not exist in $H$ or $K$.
\end{proof}

\noindent\emph{Gloss.} prune/split give subsets; meet/difference only sometimes do.

\begin{lemma}[Idempotence]
	For any $H$, $H\sqcup H=H$ and $H\sqcap H=H$.
\end{lemma}

\begin{proof}
	Identical elements are ignored by merge/meet; nothing new is created.
\end{proof}

\noindent\emph{Gloss.} Re-applying an operator to the same input does nothing.

\begin{lemma}[Monotonicity (Restricted)]
	(a) If $S\subseteq S'$, then $(H\ominus S')\sqsubseteq(H\ominus S)$.  
	(b) If $b$ refines $b'$, then $\pi_b(H)\sqsubseteq\pi_{b'}(H)$.  
	(c) $\sqcup$, $\sqcap$, and $/$ are not monotone.
\end{lemma}

\begin{proof}
	(a) Removing more elements yields fewer survivors.  
	(b) A narrower boundary yields a smaller projection.  
	(c) merge may introduce conflicts; meet/difference may construct new hypersimplices.
\end{proof}

\noindent\emph{Gloss.} prune and split are monotone; merge/meet/difference are not.

\begin{lemma}[Identity Elements]
	For $\emptyset$: $H\sqcup\emptyset=H$, $H\sqcap\emptyset=\emptyset$, $H/\emptyset=H$, $H/H=\emptyset$, and $H\ominus\emptyset=H$. For split, $\pi_{\top}(H)=H$.
\end{lemma}

\begin{proof}
	Decision tables show identity behaviours directly.
\end{proof}

\noindent\emph{Gloss.} Empty inputs leave the hypernetwork unchanged.

\begin{lemma}[Non-commutativity and Non-associativity]
	In general, $H_1\sqcup H_2\neq H_2\sqcup H_1$ and $(H_1\sqcup H_2)\sqcup H_3\neq H_1\sqcup(H_2\sqcup H_3)$, with analogous results for other operators.
\end{lemma}

\begin{proof}
	Execution order affects unification and propagation. Operators are deterministic for a fixed order but not invariant under permutation.
\end{proof}

\noindent\emph{Gloss.} Order matters but results remain valid.

\bigskip
\noindent\textbf{Meta-summary.} These lemmas show that hypernetworks are robust under construction, projection, merging, pruning, meeting, differencing, and splitting. As long as axioms~A1--A5 and conditions~C1--C7 hold, no valid hypernetwork can be made invalid. This provides the mathematical basis for the mechanisation and reliability of the structural operator algebra.

\section*{Appendix B: Operator Decision Tables and Algorithms}
\addcontentsline{toc}{section}{Appendix B: Operator Decision Tables and Algorithms}

This appendix contains the complete decision tables, auxiliary predicates, and procedural realisations of the structural operators of Hypernetwork Theory (HT). The tables serve as executable specifications: each enumerates admissible cases and their outcomes at element level, while the algorithms act on whole hypernetworks. All operators respect axioms~A1--A5 and conditions~C1--C7, preserving identity, typing, relation binding, boundary information, and modeller-supplied exclusion. Consistent with A2, no operator ever generates anti-vertices except when explicitly required by the decision tables.

\subsection*{B.1 Auxiliary Predicates}

The following predicates are used throughout the operator specifications.
\begin{itemize}[leftmargin=*]
	\item $\EqV(v,v')$ --- vertices are identical by name.
	\item $\EqHs(\varsigma,\varsigma')$ --- hypersimplices are identical in $R$, type, role order, participants, and boundary tags.
	\item $\text{roles\_compatible}(\varsigma,\varsigma')$ --- arities coincide and roles can be aligned.
	\item $\text{nonempty\_inter}(\varsigma_1,\varsigma_2)$ --- participant intersection is non-empty.
	\item $\text{empty\_inter}(\varsigma_1,\varsigma_2)$ --- participant intersection is empty.
	\item $\text{orphan}(v)$, $\text{orphan}(\tilde v)$ --- vertex or anti-vertex has no incident hypersimplices.
	\item $\text{deleted}(p)$ --- participant $p$ has been removed in the current pruning sweep.
	\item $\text{wellformed}(\varsigma)$ --- hypersimplex satisfies A3--A5 (typing, relation binding, boundary consistency).
	\item $b(e)$ --- element $e$ carries boundary tag $b$.
	\item $\text{desc}(b)$ --- descendants reached by a percolating boundary.
	\item $\text{parts}(\varsigma)$ --- participants added to closure during split.
	\item $\text{identical\_in\_}H(e)$ --- element $e$ appears in $H$ with identical $R$, type, roles, and boundary tags.
\end{itemize}

\subsection*{B.2 Merge ($\sqcup$)}

\begin{table}[H]
	\centering
	\caption{Decision table for merge (candidate $e \in H_2$ relative to $H_1$).}
	\label{tab:merge-precedence}
	\begin{tabular}{p{5.1cm}p{7.2cm}p{2cm}}
		\toprule
		\textbf{Condition} & \textbf{Notes} & \textbf{Outcome} \\
		\midrule
		$e = \tilde v$ & Explicit modeller-supplied exclusion (A2) & Insert \\
		$e=v \wedge \exists v':\EqV(v,v')$ & Duplicate vertex; idempotent & Ignore \\
		$e=v \wedge \exists v':\neg\EqV(v,v')$ & Name clash (A1) & Conflict \\
		$e=\varsigma \wedge \exists \varsigma':\EqHs(\varsigma,\varsigma')$ & Identical hypersimplex & Ignore \\
		$e=\varsigma \wedge \exists \varsigma':\neg\EqHs(\varsigma,\varsigma')$ & SameName conflict & Conflict \\
		$\text{type}(\varsigma)\neq\text{type}(\varsigma')$ & Aggregation type mismatch (A3) & Conflict \\
		$R(\varsigma)\neq R(\varsigma')$ & Relation symbol mismatch (A4) & Conflict \\
		$\neg\text{roles\_compatible}(\varsigma,\varsigma')$ & Arity/role mismatch & Conflict \\
		$e=v \wedge v\notin H_1$ & New vertex & Insert \\
		$e=\varsigma \wedge \varsigma\notin H_1$ & New hypersimplex; preserve tags & Insert \\
		$\text{roles\_compatible}(\varsigma,\varsigma')$ & Compatible union & Unify \\
		\bottomrule
	\end{tabular}
\end{table}

\subsection*{B.3 Meet ($\sqcap$)}

\begin{table}[H]
	\centering
	\caption{Decision table for meet (pairs $e_1\in H_1$, $e_2\in H_2$).}
	\label{tab:meet-precedence}
	\begin{tabular}{>{\raggedright\arraybackslash}p{5.1cm}p{7.2cm}p{2cm}}
		\toprule
		\textbf{Condition} & \textbf{Notes} & \textbf{Outcome} \\
		\midrule
		$v_1 = v_2$ & Identical vertices & Retain \\
		$v_1 \neq v_2$ & Different identifiers & Drop \\
		$\tilde v_1 = \tilde v_2$ & Matching exclusion markers & Retain \\
		$\tilde v_1 \neq \tilde v_2$ & Non-matching exclusion markers & Drop \\
		$\neg\EqHs(\varsigma_1,\varsigma_2)$ & No SameName candidate & Null \\
		$\text{type}, R, \text{roles match} \wedge \text{nonempty\_inter}$ & Align and retain overlap & Overlap \\
		$\text{type}=\alpha \wedge \text{empty\_inter}$ & Conjunction empty & Null \\
		$\text{type}=\beta \wedge \text{empty\_inter}$ & Alternatives disjoint & Null \\
		\bottomrule
	\end{tabular}
\end{table}

\subsection*{B.4 Difference ($/$)}

\begin{table}[H]
	\centering
	\caption{Decision table for difference ($e \in H_1$ relative to $H_2$).}
	\label{tab:difference-precedence}
	\begin{tabular}{p{5.1cm}p{7cm}p{2.2cm}}
		\toprule
		\textbf{Condition} & \textbf{Notes} & \textbf{Outcome} \\
		\midrule
		$v\notin H_2$ & Vertex absent from $H_2$ & Retain \\
		$v\in H_2$ & Duplicate vertex & Drop \\
		$\tilde v\notin H_2$ & Exclusion marker absent from $H_2$ & Retain \\
		$\tilde v\in H_2$ & Duplicate exclusion marker & Drop \\
		$\neg\exists\varsigma':\EqHs(\varsigma,\varsigma')$ & No SameName match & Retain \\
		$\exists\varsigma':\text{incompatible}$ & Incomparable hypersimplex & Retain \\
		$\exists\varsigma':\text{compatible}$ & Rolewise subtraction & Subtract/Drop \\
		\bottomrule
	\end{tabular}
\end{table}

\subsection*{B.5 Prune ($\ominus$)}

\begin{table}[H]
	\centering
	\caption{Decision table for prune (fixpoint sweep).}
	\label{tab:prune-precedence}
	\begin{tabular}{p{5.1cm}p{7.2cm}p{2cm}}
		\toprule
		\textbf{Condition} & \textbf{Notes} & \textbf{Outcome} \\
		\midrule
		$v \in S$ & Vertex selected for pruning & Delete \\
		$\tilde v \in S$ & Anti-vertex selected for pruning & Delete \\
		$\text{orphan}(v)$ & No incident hypersimplices & Delete \\
		$\text{orphan}(\tilde v)$ & No incident hypersimplices & Delete \\
		$\varsigma \in S$ & Hypersimplex selected for pruning & Delete \\
		$\text{orphan}(\varsigma)$ & No supporting structure & Delete \\
		$\exists p:\text{deleted}(p)$ & Dangling role after deletions & Delete \\
		$R(\varsigma)\in S$ & Prune all instances of relation $R$ & Delete all \\
		$b(e)\in S$ & Prune all elements carrying boundary $b$ & Delete all \\
		$\neg\text{wellformed}(\varsigma)$ & Violates A3--A5 & Delete \\
		survivors keep boundary tags & Tags preserved on all kept elements & Keep \\
		\bottomrule
	\end{tabular}
\end{table}

\subsection*{B.6 Split ($\pi$)}

\begin{table}[H]
	\centering
	\caption{Decision table for split (boundary- or vertex-based projection).}
	\label{tab:split-precedence}
	\begin{tabular}{p{5.1cm}p{7.2cm}p{2cm}}
		\toprule
		\textbf{Condition} & \textbf{Notes} & \textbf{Outcome} \\
		\midrule
		$b(e)$ and non-percolating & Tagged element only & Retain \\
		$e\in\text{desc}(b)$ & Percolated to descendants & Retain \\
		$v\in S$ & Seed vertex for closure & Retain $v$ \\
		$\varsigma\cap S\neq\emptyset$ & Closure under roles and typing & Retain; expand $S$ \\
		$\varsigma\cap S=\emptyset\wedge\neg b(e)$ & Outside scope & Drop \\
		retained element & Preserve tags & Keep \\
		\bottomrule
	\end{tabular}
\end{table}

\subsection*{B.7 Sub-hypernetwork Relation ($\sqsubseteq$)}

\begin{table}[H]
	\centering
	\caption{Decision table for sub-hypernetwork relation.}
	\label{tab:subhn-precedence}
	\begin{tabular}{p{5.1cm}p{7.2cm}p{2cm}}
		\toprule
		\textbf{Condition} & \textbf{Notes} & \textbf{Outcome} \\
		\midrule
		$v\notin H$ or $\tilde v\notin H$ & Missing atomic element & Fail \\
		$\varsigma\notin H$ & Missing hypersimplex & Fail \\
		$\text{identical\_in\_}H(e)$ for all $e\in H'$ & Same $R$, type, roles, tags & True \\
		\bottomrule
	\end{tabular}
\end{table}

\subsection*{B.8 Operator Algorithms}

This subsection presents the procedural realisations of the structural operators defined in Section~\ref{sec:structure-and-composition}. Each operator is executed by iterating through elements in a fixed order, applying the relevant decision table, and enforcing axioms~A1--A5 and conditions~C1--C7 after each step. This ensures that atomic identifiers exist before they participate in aggregations, that typing is fixed before role assignment, and that operator execution is deterministic: for any given ordered input, the same output is always produced.

\medskip
\paragraph{How to read this subsection.}
Each algorithm realises the semantics of a structural operator procedurally. The pseudocode follows the decision tables in Section~\ref{sec:structure-and-composition}, with comments and glosses explaining the intuition. At the core is the \textsc{Insert} routine, which embodies the merge decision table procedurally. Construction is repeated \textsc{Insert} into an empty hypernetwork, while merge applies \textsc{Insert} across two hypernetworks. All other operators follow the same skeleton but substitute their own decision logic for \textsc{InsertOrUnify}. In this sense, \textsc{Insert} is the universal primitive of structural composition.
\\

\begin{algorithm}[H]
	\caption{\textsc{Insert}$(H,e)$ \newline \emph{Gloss: Insert adds $e$ if new, ignores if identical, and replaces name-clash mismatches with a \SameName conflict hypersimplex; closure and \texttt{partOf} links are preserved.}}
	\KwPre{$H$ valid under A1--A5}
	\KwPost{$H'$ valid under A1--A5, C1--C7}
	\If{$e$ is a vertex}{
		\If{a vertex $v'$ with the same identifier exists in $H$}{
			\If{$\EqV(e,v')$}{\Return $H$ \tcp*{identical; \textsc{Ignore}}}
			\Else{replace $v'$ by \textsc{ConflictHS}(\SameName,;$v'$, $e$) \tcp*{A1: resolve by conflict replacement}}
		}
		\Else{insert $e$}
	}
	\If{$e$ is an anti-vertex}{
		insert $e$ \tcp*{A2: explicit exclusion elements; may appear from the modeller or from operators, but are never created implicitly here}
	}
	\If{$e$ is a hypersimplex $\varsigma$}{
		\textbf{Step 0 (name clash handling):}
		\If{a hypersimplex $\varsigma'$ with the same identifier exists}{
			\If{$\EqHs(\varsigma,\varsigma')$}{\Return $H$ \tcp*{identical; \textsc{Ignore}}}
			\Else{replace $\varsigma'$ by \textsc{ConflictHS}(\SameName,;$\varsigma'$, $\varsigma$); \Return $H$}
		}
		\If{a bare vertex $v'$ with the same identifier exists}{
			\If{typing/binding checks (A3--A5) succeed}{replace $v'$ by $\varsigma$ \tcp*{promotion}}
			\Else{replace $v'$ by \textsc{ConflictHS}(\SameName,;$v'$, $\varsigma$); \Return $H$}
		}
		\textbf{Step 1 (well-formedness):} verify all participants exist or are anti-vertices (A1--A2);\\
		\textbf{Step 2 (typing):} check $\alpha/\beta$ declared and immutable (A3);\\
		\textbf{Step 3 (relation/roles/boundary):} check $R$ symbol, role order/arity, and boundary scope (A4--A5);\\
		\If{all checks pass}{
			insert $\varsigma$ and establish reciprocal \texttt{partOf} links
		}
		\Else{replace (or create) by \textsc{ConflictHS}(\SameName,;failing\_element, $\varsigma$)}
	}
	enforce closure (A1--A5, C1--C7);\\
	\Return $H'$
\end{algorithm}

\bigskip
\begin{algorithm}[H]
	\caption{\textsc{Merge}$(H_1,H_2)$ \newline \emph{Gloss: Merge integrates all elements of $H_2$ into $H_1$ by repeated insertion; identicals are ignored, mismatches become \SameName conflicts.}}
	\Init{$H' \gets H_1$}
	\ForEach{$e \in H_2$ in sequence}{
		$H' \gets \textsc{Insert}(H',e)$;
	}
	\Return $H'$
\end{algorithm}

\bigskip
\begin{algorithm}[H]
	\caption{\textsc{Meet}$(H_1,H_2)$ \newline \emph{Gloss: Meet extracts the structure common to both hypernetworks, constructing $\varsigma^\ast$ where rolewise intersections are non-empty; $\bot$ means no element is produced.}}
	\Init{$H' \gets \emptyset$}
	\ForEach{$(e_1,e_2) \in H_1 \times H_2$ in sequence}{
		apply meet decision table (Table~\ref{tab:meet-precedence}) to $(e_1,e_2)$;\\
		\If{result is a vertex or anti-vertex}{add to $H'$}
		\ElseIf{result is a hypersimplex $\varsigma^\ast$}{add $\varsigma^\ast$ to $H'$}
		\ElseIf{result is $\bot$}{\textbf{continue}}
	}
	enforce closure (A1--A5, C1--C7);\\
	\Return $H'$
\end{algorithm}

\bigskip
\begin{algorithm}[H]
	\caption{\textsc{Difference}$(H_1,H_2)$ \newline \emph{Gloss: Difference keeps only those elements unique to $H_1$; for hypersimplices it performs rolewise subtraction and drops empties.}}
	\Init{$H' \gets \emptyset$}
	\ForEach{$e \in H_1$ in sequence}{
		apply difference decision table (Table~\ref{tab:difference-precedence}) to $e$ w.r.t.\ $H_2$;\\
		\If{result is a vertex or anti-vertex}{add to $H'$}
		\ElseIf{result is a hypersimplex $\varsigma^{-}$}{add $\varsigma^{-}$ to $H'$ if non-empty; otherwise skip}
	}
	enforce closure (A1--A5, C1--C7);\\
	\Return $H'$
\end{algorithm}

\bigskip
\begin{algorithm}[H]
	\caption{\textsc{Prune}$(H,S)$ \newline
		\emph{Gloss: Prune replaces selected participants by anti-vertices everywhere, then repeatedly deletes invalid or orphaned structure to a bounded fixpoint; reciprocal \texttt{partOf} links are updated.}}
	\Repeat{no further changes}{
		\ForEach{$\varsigma \in H$ such that $\varsigma$ is a hypersimplex}{
			\ForEach{participant $p$ of $\varsigma$}{
				\If{$p \in S$}{
					let $\tilde p$ be the anti-vertex for $p$ (create if needed);\\
					replace $p$ by $\tilde p$ in $\varsigma$;\\
					remove \texttt{partOf}$(p,\varsigma)$;\\
					add \texttt{partOf}$(\tilde p,\varsigma)$;
				}
			}
			apply prune decision table (Table~\ref{tab:prune-precedence}) to $\varsigma$;\\
			\If{$\varsigma$ deleted}{
				remove all \texttt{partOf} references to $\varsigma$;
			}
		}
		\ForEach{$e \in H$ such that $e$ is a vertex or anti-vertex}{
			apply prune decision table (Table~\ref{tab:prune-precedence}) to $e$;\\
			\If{$e$ deleted}{
				remove all \texttt{partOf} references to $e$;
			}
		}
		\ForEach{$\varsigma \in H$ such that $\varsigma$ is a hypersimplex}{
			\If{$\varsigma$ has no incoming \texttt{partOf} links and $\varsigma \notin S$}{
				delete $\varsigma$ from $H$;\\
				remove all \texttt{partOf} references to $\varsigma$;
			}
		}
		enforce closure (A1--A5, C1--C7);
	}
	\Return{$H$}
\end{algorithm}

\bigskip
\begin{algorithm}[H]
	\caption{\textsc{Split}$(H,C)$ \newline \emph{Gloss: Split projects a sub-hypernetwork by boundary tag or by closing over a chosen vertex set under $R$ and $\alpha/\beta$.}}
	\If{$C$ is a boundary $b$}{
		\Return $B(H,b)$ \tcp*{boundary projection equals $\pi_b(H)$}
	}
	\ElseIf{$C$ is a vertex seed set $S$}{
		initialise $S' \gets S$;\\
		\While{exists $\varsigma \in H$ touching $S'$ respecting $R$ and $\alpha/\beta$}{
			add $\varsigma$ and its participants to $S'$
		}
		\Return sub-hypernetwork induced by $S'$, preserving boundary tags
	}
\end{algorithm}

\bigskip
\noindent Together these tables and algorithms show how semantics-preserving composition is executed deterministically. \textsc{Insert} provides the universal skeleton; merge generalises construction; meet, difference, and prune specialise by substituting their own decision tables; and split realises projection. Collectively they demonstrate that the operator algebra of HT is both tractable and mechanisable: rigorous in principle and executable in practice.

\bibliographystyle{plain}
\begin{small}
	\bibliography{HypernetworkTheory.bib}
\end{small}
	
\end{document}